\documentclass{article}
\usepackage{latexsym,amsthm,amsmath,amssymb,url}
\usepackage{enumerate}
\usepackage{boxedminipage}

\newtheorem{lemma}{Lemma}

\newtheorem{definition}{Definition}
\newtheorem{theorem}{Theorem}

\newtheorem{proposition}{Proposition}

\newcommand{\tw}{{\mathbf{tw}}}

\usepackage{vmargin}
\setmarginsrb{1.2in}{1.2in}{1.2in}{1.2in}{0pt}{0pt}{0pt}{6mm}

\newcommand{\sat}{\text{\normalfont sat}}
\newcommand{\nksat}{$(\nu(F)+k)$-{\sc SAT}}
\newtheorem{krule}{Reduction Rule}

\newtheorem{brule}{Branching Rule}

\begin{document}

\title{Fixed-parameter tractability of satisfying beyond the number of variables\footnote{A preliminary version of this paper appeared in SAT 2012, Lect. Notes Comput. Sci. 7317 (2012), 341--354.}}

\author{Robert Crowston\thanks{Royal Holloway, University of London, Egham,
Surrey, UK} \and Gregory Gutin\footnotemark[2] \and Mark Jones\footnotemark[2] \and Venkatesh Raman\thanks{The Institute of Mathematical Sciences,
Chennai 600 113, India} \and Saket Saurabh\footnotemark[3] \and Anders Yeo\thanks{University of Johannesburg,
Auckland Park, 2006 South Africa}}



\date{}
\maketitle

\begin{abstract}
We consider a CNF formula $F$ as a multiset of clauses: $F=\{c_1,\ldots, c_m\}$. The set of variables of $F$ will be denoted by $V(F)$.
Let $B_F$ denote the bipartite graph with partite sets $V(F)$ and $F$ and with an edge between $v \in V(F)$ and $c \in F$ if $v \in c$ or $\bar{v} \in c$. The matching number $\nu(F)$ of $F$ is the size of a maximum matching in $B_F$. In our main result, we prove that the following parameterization of {\sc MaxSat} (denoted by $(\nu(F)+k)$-\textsc{SAT}) is fixed-parameter tractable: Given a formula $F$, decide whether we can satisfy at least $\nu(F)+k$ clauses in $F$, where $k$ is the parameter.

A formula $F$ is called variable-matched if $\nu(F)=|V(F)|.$ Let $\delta(F)=|F|-|V(F)|$ and $\delta^*(F)=\max_{F'\subseteq F} \delta(F').$
Our main result implies fixed-parameter tractability of {\sc MaxSat} parameterized by $\delta(F)$ for
variable-matched formulas $F$; this complements related results of Kullmann (2000) and Szeider (2004) for {\sc MaxSat} parameterized by $\delta^*(F)$.


To obtain our main result, we reduce $(\nu(F)+k)$-\textsc{SAT} into the following parameterization of the {\sc Hitting Set} problem (denoted by $(m-k)$-{\sc Hitting Set}): given a collection $\cal C$ of $m$ subsets of a ground set $U$ of $n$ elements, decide whether there is $X\subseteq U$ such that $C\cap X\neq \emptyset$ for each $C\in \cal C$ and $|X|\le m-k,$ where $k$ is the parameter. Gutin, Jones and Yeo (2011) proved that $(m-k)$-{\sc Hitting Set} is fixed-parameter tractable by obtaining an exponential kernel for the problem. We obtain two algorithms for $(m-k)$-{\sc Hitting Set}:
a deterministic algorithm of runtime $O((2e)^{2k+O(\log^2 k)} (m+n)^{O(1)})$ and a randomized algorithm of expected runtime $O(8^{k+O(\sqrt{k})} (m+n)^{O(1)})$. Our deterministic algorithm improves an algorithm that follows from the kernelization result of Gutin, Jones and Yeo (2011).
\end{abstract}

%
\section{Introduction}

In this paper we study a parameterization of {\sc MaxSat}. We consider a CNF formula $F$ as a multiset of clauses: $F=\{c_1,\ldots, c_m\}$. (We allow repetition of clauses.) We assume that no clause contains both a variable and its negation, and no clause is empty. The set of variables of $F$ will be denoted by $V(F)$,
and for a clause $c$, $V(c)=V(\{c\}).$
A {\em truth assignment} is a function $\tau : V(F) \rightarrow \{ \textsc{true, false} \}$.
A truth assignment $\tau$ \emph{satisfies} a clause $C$ if there exists $x \in V(F)$ such that $x \in C$ and
$\tau(x)=$ \textsc{true}, or $\bar{x} \in C$ and $\tau(x)=$  \textsc{false}.
We will denote the number of clauses in $F$ satisfied by $\tau$ as ${\rm sat}_\tau(F)$
and the maximum value of ${\rm sat}_\tau(F)$, over all $\tau$, as ${\rm sat}(F)$.

Let $B_F$ denote the bipartite graph with partite sets $V(F)$ and $F$ with an edge between $v \in V(F)$ and $c \in F$ if $v \in V(c)$.
The {\em matching number} $\nu(F)$ of $F$ is the size of a maximum matching in $B_F$.
Clearly, ${\rm sat}(F) \ge \nu(F)$ and this lower bound for ${\rm sat}(F)$ is tight as there are formulas $F$ for which ${\rm sat}(F) = \nu(F).$

In this paper we study the following parameterized problem, where the parameterization is above a tight lower bound.

\begin{center}
\fbox{~\begin{minipage}{0.9\textwidth}
$(\nu(F)+k)$-\textsc{SAT}

\emph{Instance}: A CNF formula $F$ and a positive integer $\alpha$.


\emph{Parameter}: $k=\alpha -\nu(F)$.


\emph{Question}: Is $\sat(F) \ge \alpha $?
\end{minipage}~}
\end{center}

A natural and well-studied parameter in most optimization problems is the size of the solution. In particular, for {\sc MaxSat}, the standard parameterized problem is whether $\sat(F)\ge k$ for a CNF formula $F$. Using a simple observation that $\sat(F)\ge m/2$ for every CNF formula $F$ on $m$ clauses, Mahajan and Raman \cite{MR99} showed that this problem is fixed-parameter tractable. The tight bound $\sat(F)\ge m/2$ on $\sat(F)$ means that the problem is interesting only when $k > m/2$, i.e., when the values of $k$ are relatively large. To remedy this situation, Mahajan and Raman introduced, and showed fixed-parameter tractable, a more natural parameterized problem: whether the given CNF formula has an assignment satisfying at least $m/2+k$ clauses.
Since this pioneering paper \cite{MR99}, researchers have studied numerous problems parameterized above tight bounds including a few such parameterizations of {\sc MaxSat} \cite{AloGutKimSzeYeo11,CroGutJonYeo12,GutJonYeo11}, all stated in or inspired by Mahajan {\em et al.} \cite{MahajanRamanSikdar09}.
Like the parameterizations in \cite{AloGutKimSzeYeo11,CroGutJonYeo12,GutJonYeo11}, $(\nu(F)+k)$-\textsc{SAT} will be proved fixed-parameter tractable, but unlike them, $(\nu(F)+k)$-\textsc{SAT} will be shown to have no polynomial-size kernel unless coNP$\subseteq$NP/poly, which is highly unlikely  \cite{BDFH09}.

In our main result, we show that $(\nu(F)+k)$-\textsc{SAT} is fixed-parameter tractable by obtaining an algorithm with
running time
 $O((2e)^{2k + O(\log^2 k)}(n+m)^{O(1)})$,
where $e$ is the base of the natural logarithm. (We provide basic definitions on parameterized algorithms and complexity, including kernelization, in the next section.) We also develop a randomized algorithm for $(\nu(F)+k)$-\textsc{SAT} of expected runtime $O(8^{k+O(\sqrt{k})} (m+n)^{O(1)})$.

The {\em deficiency} $\delta(F)$ of a formula $F$ is $|F|-|V(F)|$; the {\em maximum deficiency} $\delta^*(F)=\max_{F'\subseteq F} \delta(F').$
A formula $F$ is called {\em variable-matched} if $\nu(F)=|V(F)|.$ Our main result implies fixed-parameter tractability of {\sc MaxSat} parameterized by  $\delta(F)$ for variable-matched formulas $F$.

There are two related results: Kullmann \cite{Kullmann00} obtained an $O(n^{O(\delta^*(F))})$-time algorithm for solving {\sc MaxSat} for formulas $F$ with $n$ variables and Szeider \cite{Szeider04} gave an $O(f(\delta^*(F))n^4)$-time algorithm for the problem, where $f$ is a function depending on $\delta^*(F)$ only.
Note that we cannot just drop the condition of being variable-matched from our result and expect a similar algorithm: it is not hard to see that the satisfiability problem remains NP-complete for formulas $F$ with $\delta(F)=0$.

A formula $F$ is {\em minimal unsatisfiable} if it is unsatisfiable but $F\setminus c$ is satisfiable for every clause $c\in F$.
Papadimitriou and Wolfe~\cite{PapadimitriouW88} showed that recognition of minimal unsatisfiable CNF formulas is complete for the complexity class\footnote{$D^P$ is the class of problems that can be considered as the difference of two NP-problems; clearly $D^P$  contains all NP
and all co-NP problems} $D^P$.
Kleine B\"{u}ning~\cite{Buning00} conjectured that for a fixed integer $k$, it can be decided in polynomial time whether a formula $F$
with $\delta(F)\leq k$ is minimal unsatisfiable. Independently,
Kullmann~\cite{Kullmann00} and  Fleischner and Szeider~\cite{ECCC-TR00-049} (see also \cite{FleischnerKS02}) resolved this conjecture by showing that
minimal unsatisfiable formulas with $n$ variables and $n+k$ clauses can be recognized
in $n^{O(k)}$ time. Later,  Szeider~\cite{Szeider04} showed that the problem is fixed-parameter tractable by obtaining an algorithm of running time $O(2^kn^4)$.
Note that Szeider's results follow from his results mentioned in the previous paragraph and the well-known fact that $\delta^*(F)=\delta(F)$ holds for every minimal unsatisfiable formula $F$. Since every minimal unsatisfiable formula is variable-matched \cite{AhaLin86}, our main result also implies
fixed-parameter tractability of recognizing minimal unsatisfiable formula with $n$ variables and $n+k$ clauses, parameterized by $k.$




To obtain our main result, we introduce some reduction rules and branching steps and reduce the problem to a
parameterized version of {\sc Hitting Set}, namely, {\sc $(m-k)$-Hitting Set} defined below. Let $H$ be a hypergraph. A set $S\subseteq V(H)$ is called a {\em hitting set} if $e\cap S \neq \emptyset$ for all $e\in E(H)$.

\begin{center}
\fbox{~\begin{minipage}{.9\textwidth}
$(m-k)$-\textsc{Hitting Set}

\emph{Instance}: A  hypergraph $H$ ($n=|V(H)|,\ m=|E(H)|)$ and a positive  integer $k$.


\emph{Parameter}: $k$.


\emph{Question}: Does there exist a hitting set $S\subseteq V(H)$ of size $m-k$?
\end{minipage}~}
\end{center}

Gutin {\em et al.}~\cite{GutinJY11} showed that $(m-k)$-\textsc{Hitting Set} is fixed-parameter tractable by obtaining a kernel for the problem. The kernel result immediately implies a $2^{O(k^2)} (m+n)^{O(1)}$-time algorithm for the problem.
Here we obtain a faster algorithm for this problem that runs in
$O((2e)^{2k + O(\log^2 k)} (m+n)^{O(1)})$
time using the color-coding technique. This happens to be the dominating step for solving the $(\nu(F)+k)$-\textsc{SAT} problem.
We also obtain a randomized algorithm for $(m-k)$-\textsc{Hitting Set} of expected runtime $O(8^{k+O(\sqrt{k})} (m+n)^{O(1)})$.
To obtain the randomized algorithm, we reduce $(m-k)$-\textsc{Hitting Set} into a special case of the {\sc Subgraph Isomorphism} problem and use
a recent randomized algorithm of Fomin {\em et al.} \cite{FominLRSR12} for {\sc Subgraph Isomorphism}.

It was shown in \cite{GutinJY11} that the
$(m-k)$-\textsc{Hitting Set} problem cannot have a kernel whose size is polynomial in $k$ unless NP $\subseteq$ coNP/poly.
In this paper, we give a parameter preserving reduction from this problem to the $(\nu(F)+k)$-\textsc{SAT} problem, thereby showing that
$(\nu(F)+k)$-\textsc{SAT} problem has no polynomial-size kernel unless NP $\subseteq$ coNP/poly.

\paragraph{\bf Organization of the rest of the paper.} In Section \ref{sec:prelim}, we provide additional terminology and notation and some preliminary results. In Section \ref{sec:preproc}, we give a sequence of polynomial time preprocessing rules on the given input of $(\nu(F)+k)$-\textsc{SAT} and justify their correctness. In Section \ref{sec:branch}, we give two simple branching rules and reduce the resulting input to a
$(m-k)$-\textsc{Hitting Set} problem instance. Section \ref{sec:m-khs} gives an improved fixed-parameter algorithm for $(m-k)$-\textsc{Hitting Set}
using color coding. There we also obtain a faster randomized algorithm for $(m-k)$-\textsc{Hitting Set}.
Section \ref{sec:complete} summarizes the entire algorithm for the $(\nu(F)+k)$-\textsc{SAT} problem, shows its correctness and analyzes its running time. Section \ref{sec:nokern} proves the hardness of kernelization result. Section \ref{sec: concl} concludes with some remarks.

\section{Additional Terminology, Notation and Preliminaries}
\label{sec:prelim}

\noindent{\bf Graphs and Hypergraphs.} 
For a subset $X$ of vertices of a graph $G$, $N_G(X)$ denotes the set of all neighbors of vertices in $X$. When $G$ is clear from the context, we write $N(X)$ instead of $N_G(X)$. A matching {\em saturates} all end-vertices of its edges.
For a bipartite graph $G=(V_1, V_2;E)$, the classical Hall's matching theorem states that $G$ has a matching that saturates every vertex of $V_1$ if and only if $|N(X)| \geq |X|$ for every subset $X$ of $V_1$.
The next lemma follows from Hall's matching theorem: add $d$ vertices to $V_2$, each adjacent to every vertex in $V_1$.

\begin{lemma}\label{lem:dless}
Let $G=(V_1,V_2;E)$ be a bipartite graph, and suppose that
for all subsets $X\subseteq V_1$, $|N(X)|\geq |X|-d$ for some $d\geq 0$. Then $\nu(G)\ge |V_1|-d.$
\end{lemma}

We say that a bipartite graph $G=(A,B;E)$ is {\em $q$-expanding}
if for all $A'\subseteq A$, $|N_G(A')|\geq |A'|+q$. Given a matching $M$, an {\em alternating path} is a path in
which the edges belong alternatively to $M$ and not to $M$.

A {\em hypergraph} $H = (V(H), {\cal F})$ consists of a nonempty set $V(H)$ of {\em vertices} and a family $\cal F$ of nonempty subsets of
$V$ called {\em edges} of $H$ (${\cal F}$ is often denoted $E(H)$). Note that $\cal F$ may have {\em parallel} edges, i.e., copies of the same subset
of $V(H).$ For any vertex $v \in V(H)$, and any ${\cal E} \subseteq {\cal F}$, ${\cal E}[v]$ is the set of edges in
${\cal E}$ containing $v$, $N[v]$ is the set of all vertices contained in edges of ${\cal F}[v]$, and
the {\em degree} of $v$ is $d(v) = |{\cal F}[v]|$.  For a subset $T$ of vertices, ${\cal F}[T]=\bigcup_{v\in T}{\cal F}[v].$

\medskip

\noindent{\bf CNF formulas.}
For a subset $X$ of the variables of CNF formula $F$, $F_X$ denotes the subset of $F$ consisting of all clauses $c$ such that $V(c)\cap X\neq \emptyset.$ A formula $F$ is called {\em $q$-expanding} if $|X|+q\le |F_X|$ for each $X\subseteq V(F)$. Note that, by Hall's matching theorem, a formula is variable-matched if and only if it is 0-expanding. Clearly, a formula $F$ is $q$-expanding if and only if $B_F$ is $q$-expanding.

For $x\in V(F)$, $n(x)$ and $n(\bar{x})$ denote the number of clauses containing $x$ and the number of clauses containing $\bar{x},$ respectively.

A function $\pi :\ U \rightarrow \{ \textsc{true, false} \}$, where $U$ is a subset of $V(F)$, is called a {\em partial truth assignment}.
A partial truth assignment $\pi : U \rightarrow \{ \textsc{true, false} \}$ is an {\em autarky} if $\pi$ satisfies all clauses of $F_U$.
We have the following:

\begin{lemma}[\cite{CroGutJonYeo12}] \label{lem:aut}
 Let $\pi : U \rightarrow \{ \textsc{true, false} \}$ be an autarky for a CNF formula $F$ and let $\gamma$ be any truth assignment on
 $V(F)\setminus U$. Then for the combined assignment $\tau := \pi\cup \gamma$,
 it holds that
 ${\rm sat}_{\tau}(F)=|F_U|+{\rm sat}_{\gamma}(F\setminus F_U)$.
 Clearly, $\tau$ can be constructed in polynomial time given $\pi$ and $\gamma$.
\end{lemma}

Autarkies were first introduced in \cite{MS1985}; they
are the subject of much study, see, e.g., \cite{FleischnerKS02,Kul03,Szeider04}, and see \cite{BunKul09} for an overview.

\paragraph{\bf Treewidth.}  A {\em tree decomposition} of an (undirected) graph $G$ is a pair
$(U,T)$ where $T$ is a tree whose vertices we will call {\em
nodes} and $U=(\{U_{i} \mid i\in V(T)\})$ is a collection of
subsets of $V(G)$ such that
\begin{enumerate}
\item $\bigcup_{i \in V(T)} U_{i} = V(G)$,

\item for each edge $vw \in E(G)$, there is an $i\in V(T)$
such that $v,w\in U_{i}$, and

\item for each $v\in V(G)$ the set  $\{ i:\ v \in U_{i}
\}$ of nodes forms a subtree of $T$.
\end{enumerate}
The $U_i$'s are called {\em bags}.
The {\em width} of a tree decomposition $(\{ U_{i}:\ i \in V(T) \},
T)$ equals $\max_{i \in V(T)} \{|U_{i}| - 1\}$. The {\em treewidth} of
a graph $G$ is the minimum width over all tree decompositions of $G$.
We use notation $\tw(G)$ to denote the treewidth of a graph $G$.

\paragraph{\bf Parameterized Complexity.} A \emph{parameterized problem} is a subset $L\subseteq \Sigma^* \times
\mathbb{N}$ over a finite alphabet $\Sigma$. The unparameterized version of a parameterized problem $L$ is the language
$L^c = \{x \# 1^k | (x, k) \in L\}$. The problem $L$ is
\emph{fixed-parameter tractable} if the membership of an instance
$(x,k)$ in $\Sigma^* \times \mathbb{N}$ can be decided in time
$f(k)|x|^{O(1)},$ where $f$ is a function of the
{\em parameter} $k$ only~\cite{DowneyFellows99,FlumGrohe06,Niedermeier06}.
Given a parameterized problem $L$,
a \emph{kernelization of $L$} is a polynomial-time
algorithm that maps an instance $(x,k)$ to an instance $(x',k')$ (the
\emph{kernel}) such that (i)~$(x,k)\in L$ if and only if
$(x',k')\in L$, (ii)~ $k'\leq g(k)$, and (iii)~$|x'|\leq g(k)$ for some
function $g$. We call $g(k)$ the {\em size} of the kernel.
It is well-known \cite{DowneyFellows99,FlumGrohe06} that a decidable parameterized problem $L$ is fixed-parameter
tractable if and only if it has a kernel. Polynomial-size kernels are of
main interest, due to applications \cite{DowneyFellows99,FlumGrohe06,Niedermeier06}, but unfortunately not all fixed-parameter problems
have such kernels unless  coNP$\subseteq$NP/poly, see, e.g., \cite{BDFH09,BTY09,DLS09}.

\medskip

For a positive integer $q$, let $[q]=\{1,\ldots,q\}$.

\section{Preprocessing Rules}
\label{sec:preproc}

In this section we give preprocessing rules and their correctness.

Let $F$ be the given CNF formula on $n$ variables and $m$ clauses with a maximum matching $M$ on $B_F$,
the variable-clause bipartite graph corresponding to $F$. Let $\alpha$ be a given integer and recall that our goal is to
check whether $\sat(F) \ge \alpha$.
For each preprocessing rule below, we let $(F',\alpha')$ be the instance resulting by the application of the rule on $(F,\alpha)$.
We say that a rule is {\em valid} if $(F,\alpha)$ is a {\sc Yes} instance if and only if $(F',\alpha')$ a {\sc Yes} instance.

\begin{krule} \label{pureliteral}
Let $x$ be a variable such that $n(x)=0$ (respectively $n(\bar{x})=0$). Set $x=\textsc{false}$ ($x=\textsc{true}$) and remove all the clauses that contain $\bar{x}$ ($x$).
Reduce $\alpha$ by $n(\bar{x})$ (respectively $n(x)$).
\end{krule}

The proof of the following lemma is immediate.

\begin{lemma}
\label{lem:cor1}
If $n(x)=0$ (respectively $n(\bar{x})=0$) then $\sat(F) = \sat(F')+n(\bar{x})$
(respectively $\sat(F) = \sat(F')+n(x)$), and so Rule \ref{pureliteral} is valid.
\end{lemma}

\begin{krule}
\label{twooccurrence}
Let $n(x)=n(\bar{x})=1$ and let $c'$ and $c''$ be the two
clauses containing $x$ and $\bar{x}$, respectively. Let $c^* = (c' - x) \cup
(c'' - \bar{x})$ and let $F'$ be obtained from $F$ be deleting $c'$ and $c''$ and adding the clause $c^*$.
Reduce $\alpha$ by $1$.
\end{krule}

\begin{lemma}
\label{lem:cor2}
For $F$ and $F'$ in Reduction Rule \ref{twooccurrence}, $\sat(F) = \sat(F')+1$, and so Rule \ref{twooccurrence} is valid.
\end{lemma}
\begin{proof}
Consider any assignment for $F$. If it satisfies both $c'$ and $c''$, then the same assignment will satisfy $c^*$. So when restricted to variables of $F'$, it will satisfy at least $sat(F)-1$ clauses of $F'$.
Thus $sat(F') \geq sat(F) -1$ which is equivalent to $sat(F) \leq sat(F') +1$.  Similarly if an assignment $\gamma$ to $F'$ satisfies $c^*$ then at least one of $c',c''$ is satisfied by $\gamma$. Therefore by setting $x$ true if $\gamma$ satisfies $c''$ and false otherwise, we can extend $\gamma$ to an assignment on
$F$ that satisfies both of $c',c''$. On the other hand, if $c^*$ is not satisfied by $\gamma$ then neither $c'$ nor $c''$ is satisfied by $\gamma$,
and any extension of $\gamma$ will satisfy exactly one of $c',c''$. Therefore in either case  $sat(F) \geq sat(F') +1$. We conclude that $\sat(F) = \sat(F')+1$, as required.
 \end{proof}

Our next reduction rule is based on the following lemma proved in Fleischner {\em et al.} \cite[Lemma 10]{FleischnerKS02},
Kullmann \cite[Lemma 7.7]{Kul03}  and Szeider \cite[Lemma 9]{Szeider04}.

\begin{lemma}\label{lem:red}
Let $F$ be a CNF formula. Given a maximum matching in $B_F$, in time $O(|F|)$ we can find an autarky $\pi: U \rightarrow \{ \textsc{true, false} \}$ such that
$F \setminus F_U$ is 1-expanding.
\end{lemma}

\begin{krule}
\label{k1}
Find an autarky $\pi: U \rightarrow \{ \textsc{true, false} \}$ such that
$F \setminus F_U$ is 1-expanding. Set $F'=F \setminus F_U$ and reduce $\alpha$ by $|F_U|.$
\end{krule}

\noindent
The next lemma follows from Lemma \ref{lem:aut}.

\begin{lemma}
\label{lem:cor3}
For $F$ and $F'$ in Reduction Rule \ref{k1}, $\sat(F) = \sat(F')+|F_U|$ and so Rule \ref{k1} is valid.
\end{lemma}

%
%
%
%
%

After exhaustive application of Rule \ref{k1}, we may assume that the resulting formula is $1$-expanding.
For the next reduction rule, we need the following results.

\begin{theorem}[Szeider \cite{Szeider04}]
\label{n_plus_1}
Given a variable-matched formula $F$, with $|F|=|V(F)|+1$, we can decide whether $F$ is satisfiable in
time $O(|V(F)|^3)$.
\end{theorem}

Consider a bipartite graph $G=(A,B;E)$. Recall that a formula $F$ is $q$-expanding if and only if
$B_F$ is $q$-expanding. From a bipartite graph $G = (A, B; E)$, $x\in A$ and $q\geq 1$, we obtain a
bipartite graph $G_{qx}$, by adding new vertices $x_1,\ldots,x_q$ to $A$ and adding edges such that new vertices have exactly
the same neighborhood as $x$, that is, $G_{qx}=(A\cup \{x_1,\ldots,x_q\}, B; E\cup \{(x_i,y):\ (x,y)\in E \})$. The following result is well known.

\begin{lemma}{\rm \cite[Theorem 1.3.6]{LovaszP86}}\label{lem:qexp}
Let $G=(A,B; E)$ be a $0$-expanding bipartite graph. Then  $G$ is $q$-expanding if and only if $G_{qx}$
is $0$-expanding for all $x\in A$.
\end{lemma}

\begin{lemma}
\label{2exp}
Let $G=(A,B; E)$ be a 1-expanding bipartite graph. In polynomial time, we can check whether $G$ is 2-expanding, and if it is not, find a set $S\subseteq A$
such that $|N_G(S)|=|S|+1.$
\end{lemma}
\begin{proof}
Let $x\in A$. By Hall's Matching Theorem, $G_{2x}$ is $0$-expanding if and only if $\nu(G_{2x})=|A|+2.$ Since we can check the last condition in polynomial time, by Lemma \ref{lem:qexp} we can decide whether $G$ is 2-expanding in polynomial time. So, assume that $G$ is not 2-expanding and we know this because $G_{2y}$ is not $0$-expanding for some $y\in A.$ By Lemma 3(4) in \cite{Szeider04}, in polynomial time, we can find a set $T\subseteq A\cup \{y_1,y_2\}$ such that
$|N_{G_{2y}}(T)|<|T|.$ Since $G$ is 1-expanding, $y_1,y_2\in T$ and $|N_{G_{2y}}(T)|=|T|-1.$ Hence, $|S|+1=|N_G(S)|,$ where $S=T\setminus \{y_1,y_2\}$.
\end{proof}

\noindent
For a formula $F$ and a set $S\subseteq V(F)$, $F[S]$ denotes the formula obtained from $F_S$ by deleting all variables not in $S$.

\begin{krule} \label{k2}
Let $F$ be a 1-expanding formula and let $B=B_F$.  Using Lemma \ref{2exp}, check whether $F$ is 2-expanding. If it is then do not change $F$,
otherwise find a set $S \subseteq V(F)$ with $|N_{B}(S)|=|S|+1$. Let $M$ be a matching that saturates $S$ in $B[S\cup N_{B}(S)]$ (that exists as $B[S\cup N_{B}(S)]$ is $1$-expanding). Use Theorem \ref{n_plus_1} to decide whether $F[S]$ is satisfiable,
and proceed as follows.
\begin{description}
\item[\mbox{$F[S]$} is satisfiable:] Obtain a new formula $F'$ by removing all clauses in  $N_B(S)$ from $F$. Reduce
$\alpha$ by $|N_{B}(S)|$.

%

\item[\mbox{$F[S]$} is not satisfiable:]
Let $c'$ be the clause obtained by deleting all variables in $S$ from  $\cup_{c'' \in N_B(S)} c''$. That is,
a literal $l$ belongs to $c'$ if and only if it belongs to some clause in $N_B(S)$ and the variable corresponding to $l$
is not in $S$.
Obtain a new formula $F'$ by removing  all clauses in  $N_B(S)$ from $F$  and
adding $c'$. Reduce $\alpha$ by $|S|$.


\end{description}
\end{krule}
\begin{lemma}
\label{lem:cor4}
For $F$, $F'$ and $S$ introduced in Rule \ref{k2},
if \mbox{$F[S]$} is satisfiable $\sat(F) = \sat(F')+|N_{B}(S)|$, otherwise $\sat(F) = \sat(F')+|S|$ and thus Rule \ref{k2} is valid.
\end{lemma}
\begin{proof} We consider two cases.

\noindent{\bf Case 1: $F[S]$ is satisfiable.} Observe that there is an autarky on $S$ and thus by Lemma \ref{lem:aut}, $\sat(F) = \sat(F')+|N_{B}(S)|$.


\noindent{\bf Case 2: $F[S]$ is not satisfiable.}
Let $F'' = F' \setminus c'$. As any optimal truth assignment to $F$ will satisfy at least $\sat(F)- |N_B(S)|$ clauses of $F''$, it follows that $\sat (F) \leq \sat(F'') + |N_B(S)| \leq \sat(F') + |N_B(S)|$.

Let $y$ denote the clause in $N_B(S)$ that is not matched to a variable in $S$ by $M$. Let $S'$ be the set of variables, and $Z$ the set of clauses, that can be
reached from $y$ with an $M$-alternating path in $B[S \cup N_B(S)].$
We  argue now that $Z=N_B(S)$.
Since $Z$ is made up of clauses that are reachable in $B[S\cup N_{B}(S)]$ by an $M$-alternating path from the single unmatched clause $y$,
$|Z| = |S'|+1$. It follows that $|N_{B}(S)\backslash Z| = |S \backslash S'|$, and $M$ matches every clause in $N_{B}(S)\backslash Z$ with a variable in $S \backslash S'$. Furthermore,
$N_B (S \backslash S') \cap Z = \emptyset $ as otherwise the matching partners of some elements of $S \backslash S'$ would have been reachable by an $M$-alternating path from $y$, contradicting the definition of $N_B(S)$ and $S'$. Thus $S\setminus S'$ has an autarky such that $F \setminus F_{S\setminus S'}$ is $1$-expanding which would have been detected by Rule \ref{k1}, hence $S\setminus S' = \emptyset$ and so $S=S'$. That is, all clauses in $N_B(S)$ are reachable from the unmatched clause $y$ by an $M$-alternating path. We have now shown that $Z=N_B(S)$, as desired.

Suppose that there exists an assignment $\gamma$ to $F'$, that satisfies $\sat(F')$ clauses of $F'$ that also satisfies $c'$.
Then there exists a clause $c'' \in N_{B}(S)$ that is satisfied by $\gamma$. As $c''$ is reachable from $y$ by an $M$-alternating path, we can modify $M$ to include $y$ and exclude $c''$, by taking the symmetric difference of the matching and the $M$-alternating path from $y$ to $c''$. This will give a matching saturating $S$ and $N_{B}(S)\setminus c''$, and we use this matching to extend the assignment $\gamma$ to one which satisfies all of $N_{B}(S) \backslash c''$. We therefore have satisfied all the clauses of $N_{B}(S)$.  Therefore since $c'$ is satisfied in $F'$ but does not appear in $F,$ we have satisfied  extra $|N_{B}(S)|-1 = |S|$ clauses.
%
Suppose on the other hand that every assignment $\gamma$ for $F'$ that satisfies $\sat(F')$ clauses does not satisfy $c'$.
We can use the matching on $B[S\cup N_{B}(S)]$ to satisfy $|N_{B}(S)|-1$ clauses in $N_{B}(S)$, which would give us an additional $|S|$ clauses in $N_{B}(S)$.
Thus $\sat(F) \geq \sat(F') + |S|$.

 As $|N_B(S)|=|S|+1$, it suffices to show that $\sat(F) < \sat(F') + |N_B(S)|$. Suppose that there exists an assignment $\gamma$ to $F$ that satisfies $\sat(F') + |N_B(S)|$ clauses, then it must satisfy all the clauses of $N_B(S)$
and $\sat(F')$ clauses of $F''$.
 As $F[S]$ is not satisfiable, variables in $S$ alone can not satisfy all of $N_B(S)$. Hence there exists a clause $c'' \in N_B(S)$ such that there is a variable $v \in V(c'')\setminus S$ that satisfies $c''$.  But then $v \in V(c')$ and hence $c'$ would be satisfiable by $\gamma$, a contradiction as $\gamma$ satisfies $\sat(F')$ clauses of $F''$.
\end{proof}

\section{Branching Rules and Reduction to $(m-k)$-\textsc{Hitting Set}}
\label{sec:branch}
Our algorithm first applies Reduction
Rules~\ref{pureliteral}, \ref{twooccurrence}, \ref{k1} and  \ref{k2} exhaustively on $(F,\alpha)$.
Then it applies two branching rules we describe below,
in the following order.

Branching on a variable $x$ means that the algorithm constructs two instances of the problem,
one by substituting $x=\textsc{true}$ and simplifying the instance and the other by substituting $x=\textsc{false}$ and
simplifying the instance. Branching on $x$ or $y$ being false means that the algorithm constructs two instances of the problem,
one by substituting $x=\textsc{false}$ and simplifying the instance and the other by substituting $y=\textsc{false}$ and
simplifying the instance.
Simplifying an instance is done as follows. For any clause $c$, if $c$ contains a literal $z$ with $z=\textsc{true}$, remove $c$ and reduce $\alpha$ by $1$. If $c$ contains a literal $z$ with $z=\textsc{false}$ and $c$ contains other literals, remove $z$ from $c$. If $c$ consists of the single literal $z=\textsc{false}$, remove $c$.

A branching rule is correct if the instance on which it is applied is a
{\sc Yes}-instance if and only if the simplified instance of (at least) one of the branches is a {\sc Yes}-instance.



\begin{brule}
\label{bran1}
If  $n(x) \geq 2$ and $n(\bar{x}) \geq 2$ then we branch on $x$.
\end{brule}

Before attempting to apply Branching Rule \ref{bran2},
we apply the following rearranging step:
For all variables $x$ such that $n(\bar{x})=1$, swap literals $x$ and $\bar{x}$ in all clauses. Clearly, this will not change ${\rm sat}(F)$.
 Observe that now for every variable $n(x)=1$ and $n(\bar{x})\geq 2$.

\begin{brule}
\label{bran2}
If there is a clause $c$ such that positive literals $x,y\in c$
then we
branch on $x$ being false or $y$ being false.
\end{brule}

Branching Rule \ref{bran1} is exhaustive and thus its correctness also
follows. When we reach Branching Rule \ref{bran2}
for every variable $n(x)=1$ and $n(\bar{x})\geq 2$. As $n(x)=1$ and $n(y)=1$ we note that $c$ is
the only clause containing these literals. Therefore there exists an optimal solution with $x$ or $y$ being
false (if they are both true just change one of them to false). Thus, we have the following:

\begin{lemma}
\label{lem:cor5}
Branching Rules \ref{bran1} and \ref{bran2} are correct.
\end{lemma}

%
%
%
%
%

Let $(F,\alpha)$ be the given instance on which Reduction Rules~\ref{pureliteral}, \ref{twooccurrence}, \ref{k1} and  \ref{k2}, and
Branching Rules~\ref{bran1} and \ref{bran2} do not apply. Observe that for such an instance $F$
the following holds:
\begin{enumerate}
\setlength{\itemsep}{-2pt}
 \item For every variable $x$, $n(x)=1$ and $n(\bar{x})\geq 2$.
\item Every clause contains at most one positive literal.
\end{enumerate}
We call a formula $F$ satisfying the above properties {\em special}.
In what follows we describe an algorithm for our problem on special instances. Let $c(x)$ denote the {\em unique} clause
containing positive literal $x$. We can obtain a matching  saturating $V(F)$ in
$B_F$ by taking the edge connecting the variable $x$ and the clause $c(x)$.  We denote the resulting matching by $M_u$.

We first describe a transformation that will be helpful in reducing our problem to $(m-k)$-\textsc{Hitting Set}.
Given a formula $F$ we obtain a new formula $F'$ by changing the clauses of $F$
as follows. If there exists some $c(x)$ such that $|c(x)| \geq 2$, do the following.
Let $c' = c(x) - x$ (that is, $c'$ contain the same literals as $c(x)$ except for $x$)
and add $c'$ to all clauses containing the literal $\bar{x}$. Furthermore remove $c'$ from $c(x)$ (which results
in $c(x)=(x)$ and therefore $|c(x)|=1$).

Next we prove the validity of the above transformation.
\begin{lemma}
\label{lem:transformation}
 Let $F'$ be the formula obtained by applying the transformation described above on $F$.
  Then $\sat(F')=\sat(F)$ and   $\nu(B_F)=\nu(B_{F'})$.
\end{lemma}
\begin{proof}
We note that the matching  $M_u$ remains a matching in $B_{F'}$
and thus $\nu(B_F)=\nu(B_{F'})$.
Let $\gamma$ be any truth assignment to the variables in $F$ (and $F')$ and note
that if $c'$ is false under $\gamma$ then $F$ and $F'$ satisfy exactly the same clauses under $\gamma$ (as we add and subtract something
false to the clauses). So assume that $c'$ is true under $\gamma$.

If $\gamma$ maximizes the number of satisfied clauses in $F$ then clearly we may assume that $x$ is false (as $c(x)$ is true due to $c'$).
Now let $\gamma'$ be equal to $\gamma$ except the value of $x$ has been flipped to true. Note that exactly the same clauses are satisfied
in $F$ and $F'$ by $\gamma$ and $\gamma'$, respectively.
Analogously, if an assignment maximizes the number of satisfied clauses in $F'$ we may assume that $x$ is true and
by changing it to false we satisfy equally many clauses in $F$. Hence,  $\sat(F')=\sat(F)$.
\end{proof}

Given a special instance $(F,\alpha)$ we apply the above transformation repeatedly until no longer possible and obtain an instance
$(F',\alpha)$ such that $\sat(F')=\sat(F)$, $\nu(B_F)=\nu(B_{F'})$ and $|c(x)|=1$ for all $x \in V(F')$.
We call such an instance $(F',\alpha)$ {\em transformed special}. Observe that, it takes polynomial time, to obtain the transformed special
instance from a given special instance.

For simplicity of presentation we denote the transformed special
instance by $(F,\alpha)$. Let $C^*$ denote all clauses that are not matched by
$M_u$ (and therefore only contain negated literals). We associate a hypergraph $H^*$ with the transformed special instance.
Let $H^*$ be the hypergraph with
vertex set $V(F)$ and edge set $E^* = \{ V(c) ~|~ c \in C^* \}$.

%
%
%
%
%
We now show the following equivalence between $(\nu(F)+k)$-\textsc{SAT} on transformed special instances and  $(m-k)$-\textsc{Hitting Set}.

\begin{lemma}
\label{lem:newtransformation}
Let $(F,\alpha)$ be the transformed special instance and $H^*$ be the hypergraph associated with it.
Then $\sat(F)\geq \alpha$ if and only if there is a hitting set in $H^*$ of size at most $|E(H^*)|-k$, where $k=\alpha-\nu(F)$.
\end{lemma}
\begin{proof}
We start with a simple observation about an assignment satisfying the maximum number of clauses of $F$.
There exists an optimal truth assignment to $F$, such that all clauses in $C^*$ are true.
Assume that this is not the case and let $\gamma$ be an optimal truth assignment satisfying as many clauses
from $C^*$ as possible and assume that $c \in C^*$ is not satisfied. Let $\bar{x} \in c$ be an arbitrary literal and note that
$\gamma(x)=\textsc{true}$. However, changing $x$ to false does not decrease the number
of satisfied clauses in $F$
and increases the number of satisfied clauses in $C^*$.

Now we show that  $\sat(F)\geq \alpha$ if and only if there is a hitting set in $H^*$ of size at most $|E(H^*)|-k$.
Assume that $\gamma$ is an optimal truth assignment to $F$, such that all clauses in $C^*$ are true. Let $U \subseteq V(F)$
be all variables that are false in $\gamma$ and note that $U$ is a hitting set in $H^*$. Analogously if $U'$ is a hitting set in $H^*$
then by letting all variables in $U'$ be false and all other variables in $V(F)$ be true we get a truth assignment that satisfies
$|F|-|U'|$ clauses in $F$.
Therefore if $\tau(H^*)$ is the size of a minimum hitting set in $H^*$ we have $\sat(F) = |F|-\tau(H^*)$.
Hence, $\sat(F) = |F| - \tau(H^*) = |V(F)|+|C^*| - \tau(H^*)$ and thus
$\sat(F)\geq \alpha$ if and only if $|C^*|-\tau(H^*) \geq k$, which is equivalent to $\tau(H^*) \leq |E(H^*)|-k$.
\end{proof}

Therefore our problem is fixed-parameter tractable on transformed special instances, by the next theorem that follows from the kernelization result in \cite{GutinJY11}.
\begin{theorem}
\label{thm:setcover}
There exists an algorithm for  $(m-k)$-\textsc{Hitting Set} running in time $2^{O(k^2)}+O((n+m)^{O(1)})$.
\end{theorem}

In the next section we give faster algorithms for  $(\nu(F)+k)$-\textsc{SAT} on transformed special instances  by giving faster algorithms for
$(m-k)$-\textsc{Hitting Set}.

\section{Algorithms for $(m-k)$-\textsc{Hitting Set}}
\label{sec:m-khs}
To obtain faster algorithms for $(m-k)$-\textsc{Hitting Set}, we utilize the following concept of $k$-\emph{mini-hitting set} introduced in~\cite{GutinJY11}.
\begin{definition}
Let $H=(V,{\cal F})$ be a hypergraph and $k$ be a nonnegative integer. A $k$-\emph{mini-hitting set} is a set $S_{\textsc{mini}} \subseteq V$ such that $|S_{\textsc{mini}}| \le k$ and
$|{\cal F}[S_{\textsc{mini}}]| \ge |S_{\textsc{mini}}| + k$.
\end{definition}

\begin{lemma}[\cite{GutinJY11}]
\label{miniDominatingSet}
A hypergraph $H$ has a hitting set of size at most $m-k$ if and only if it has a $k$-mini-hitting set. Moreover,
given a $k$-mini-hitting set $S_{\textsc{mini}}$, we can construct a hitting set $S$ with $|S| \le m-k$ such that $S_{\textsc{mini}} \subseteq S$ in polynomial time.
\end{lemma}

\subsection{Deterministic Algorithm}

Next we give an algorithm that finds a  $k$-mini-hitting set $S_{\textsc{mini}}$ if it exists, in time $c^k(m+n)^{O(1)}$, where $c$ is a constant.
We first describe a randomized
algorithm based on color-coding~\cite{AlonYZ95} and then derandomize it using hash functions. Let $\chi~:~E(H)\rightarrow [q]$
be a function. For a subset $S\subseteq V(H)$, $\chi(S)$ denotes the maximum subset $X\subseteq [q]$ such that for all
$i \in X$ there exists an edge $e \in E(H)$ with $\chi(e)=i$ and $e\cap S\neq \emptyset$.
A subset $S\subseteq V(H)$ is called a {\em colorful hitting set} if $\chi(S)=[q].$
We now give a procedure that  given a coloring function $\chi$ finds a minimum
colorful hitting set, if it exists.  This algorithm will be useful in obtaining  a  $k$-mini-hitting set $S_{\textsc{mini}}$.
\begin{lemma}
\label{lemma:colorfulhittingset}
 Given a hypergraph $H$ and a coloring function $\chi~:~E(H)\rightarrow [q]$, we can find a minimum colorful hitting set
if there exists one in time $O(2^q q(m+n))$.
\end{lemma}
\begin{proof}
We first check whether  for every $i\in [q]$, $\chi^{-1}(i)\neq \emptyset$. If for any $i$ we have that $\chi^{-1}(i)= \emptyset$, then we return
that there is no colorful hitting set.  So we may assume that for all $i\in [q]$, $\chi^{-1}(i)\neq \emptyset$. We will give an
algorithm using dynamic programming over subsets of $[q]$. Let $\gamma$ be an array of size $2^{q}$ indexed by the subsets of $[q]$.
For a subset $X\subseteq [q]$, let $\gamma[X]$ denote the size of a smallest
set $W\subseteq V(H)$ such that $X\subseteq \chi(W)$. We obtain a recurrence for $\gamma[X]$ as follows:
\begin{equation*}
 \gamma[X] = \left\{
\begin{array}{rl}
\min_{(v\in V(H), \chi(\{v\}) \cap X \neq \emptyset )} \{1+ \gamma[X\setminus \chi(\{v\})]\} & \text{if } |X| \geq 1,\\
0 & \text{if } X = \emptyset.\\
\end{array} \right.
\end{equation*}
The correctness of the above recurrence is clear. The algorithm computes $\gamma[[q]]$ by filling the $\gamma$
in the order of increasing set sizes. Clearly, each cell can be filled in time $O(q(n+m))$ and thus the whole array can be
filled in time $O(2^qq(n+m))$. The size of a minimum colorful hitting set is given by $\gamma[[q]]$. We can obtain a minimum colorful hitting set by the routine back-tracking.
\end{proof}

Now we describe a randomized procedure to obtain a $k$-mini-hitting set $S_{\textsc{mini}}$ in a hypergraph $H$, if there exists one.
We do the following for each possible value $p$ of $|S_{\textsc{mini}}|$ (that is, for $1 \le p \le k$).
Color $E(H)$ uniformly at random with colors from $[p+k]$; we denote this random coloring by $\chi$.
Assume that there is a $k$-mini-hitting set $S_{\textsc{mini}}$ of size $p$ and some $p+k$ edges $e_1,\ldots,e_{p+k}$ such that for all
$i\in [p+k]$, $e_i\cap S_{\textsc{mini}}\neq \emptyset$.
The probability that for all $1\leq i < j \leq p+k$ we have that $\chi(e_i)\neq \chi(e_j)$ is
$\frac{(p+k)!}{(p+k)^{p+k}}\geq e^{-(p+k)}\geq e^{-2k}$. Now, using Lemma~\ref{lemma:colorfulhittingset} we can test in time $O(2^{p+k}(p+k) (m+n))$
whether there is a colorful hitting set of size at most $p$. Thus with probability at least $e^{-2k}$ we can find a $S_{\textsc{mini}}$, if
there exits one. To boost the probability we repeat the procedure $e^{2k}$ times and thus in time
$O((2e)^{2k}2k(m+n)^{O(1)})$
we find a
$S_{\textsc{mini}}$, if there exists one, with probability at least
$1- (1- \frac{1}{e^{2k} })^{e^{2k}}\geq \frac{1}{2}$. If we obtained $S_{\textsc{mini}}$ then using Lemma~\ref{miniDominatingSet} we can construct a
hitting set of $H$ of size at most $m-k$.


To derandomize the procedure, we need to replace the first step of the procedure where we color the edges of
$E(H)$ uniformly at random from the set $[p+k]$ to a deterministic one. This is done by making use of an $(m,p+k,p+k)$-{\em perfect hash family}.
 An  $(m,p+k,p+k)$-perfect hash family, ${\cal H}$, is a set of functions from $[m]$ to $[p+k]$ such that for
 every subset $S\subseteq [m]$ of size $p+k$ there exists a function $f\in {\cal H}$ such that $f$ is injective on $S$. That is, for all $i,j\in S$, $f(i)\neq f(j)$.  There exists a construction of an $(m,p+k,p+k)$-perfect hash family of size
 $O(e^{p+k} \cdot k^{O(\log k)}\cdot \log m)$ and one can produce this family in time linear in the output size~\cite{Srinivasan95}. Using
an $(m,p+k,p+k)$-perfect hash family $\cal H$ of size at most $O(e^{2k}\cdot k^{O(\log k)}\cdot \log m)$ rather than a random coloring we get the desired deterministic algorithm. To see this, it is enough to observe that if there is a subset $S_{\textrm{mini}}\subseteq V(H)$ such that $|{\cal F}[S_{\textsc{mini}}]| \ge |S_{\textsc{mini}}| + k$ then there exists a coloring $f\in {\cal H}$ such that the $p+k$ edges $e_1,\ldots,e_{p+k}$
that intersect $S_{\textrm{mini}}$  are distinctly colored. So if we generate all colorings from ${\cal H}$ we will encounter the desired $f$.
Hence for the given $f$, when we apply Lemma~\ref{lemma:colorfulhittingset} we get the desired result.
This concludes the description. The total time of the derandomized algorithm is
$O(k2^{2k} (m+n)e^{2k} \cdot k^{O(\log k)}\cdot \log m)=O((2e)^{2k + O(\log^2 k)}(m+n)^{O(1)})$.
\begin{theorem}
\label{thm:fasterm-khset}
There exists an algorithm solving  $(m-k)$-\textsc{Hitting Set} in time \\
$O((2e)^{2k + O(\log^2 k)}(m+n)^{O(1)})$.
\end{theorem}

By Theorem \ref{thm:fasterm-khset} and the transformation discussed in Section \ref{sec:branch} we have the following theorem.

\begin{theorem}\label{thm:fpt-specialinst}
There exists an algorithm solving a transformed special instance of $(\nu(F)+k)$-\textsc{SAT} in time
$O((2e)^{2k+O(\log^2 k)}(m+n)^{O(1)}).$
\end{theorem}

\subsection{Randomized Algorithm}

In this subsection we give a randomized algorithm for $(m-k)$-\textsc{Hitting Set} running in time
$O(8^{k+O(\sqrt{k})}(m+n)^{O(1)})$. However, unlike the algorithm presented in the previous subsection
we do not know how to derandomize this algorithm. Essentially, we give a randomized algorithm to find
a  $k$-mini-hitting set $S_{\textsc{mini}}$ in the hypergraph $H$, if it exists.

Towards this we introduce
notions of a star-forest and a bush. We call $K_{1,\ell}$ a {\em star of size} $\ell$; a vertex of degree $\ell$ in $K_{1,\ell}$ is a {\em central vertex}
(thus, both vertices in $K_{1,1}$ are central).
A {\em star-forest} is a forest consisting of stars. A star-forest $F$ is said to have {\em dimension} $(a_1,a_2,\ldots,a_p)$
if $F$ has $p$ stars with sizes $a_1$, $a_2$, $\ldots$, $a_p$ respectively.  Given a star-forest $F$ of
dimension $(a_1,a_2,\ldots,a_p)$, we construct a graph, which we call a {\em bush of dimension}
$(a_1,a_2,\ldots,a_p)$, by adding a triangle $(x,y,z)$ and making $y$ adjacent to a
central vertex of in every star of $F$.

For a hypergraph $H=(V,{\cal F})$, the {\em incidence
bipartite graph} $B_H$ of $H$ has partite sets $V$ and $\cal F$, and there is an edge between $v\in V$  and $e\in \cal F$ in $H$ if $v\in e$.
Given $B_H$, we construct $B_H^*$ by adding a triangle $(x,y,z)$ and making $y$
adjacent to every vertex in the $V$.  The following lemma relates $k$-mini-hitting sets to bushes.

\begin{lemma}
\label{lem:alternatechar}
A hypergraph $H=(V,{\cal F})$ has a $k$-mini-hitting set $S_{\textsc{mini}}$ if and only if there exists a tuple $(a_1,\ldots,a_p)$ such that
\begin{description}
  \item[(a)] $p\leq k$,  $a_i\geq 1$ for all $i\in [p]$, and $\sum_{i=1}^{p}a_i=p+k$; and
  \item[(b)] there exists a subgraph of $B_H^*$ isomorphic to a bush of dimension  $(a_1,\ldots,a_p)$.
\end{description}
\end{lemma}
\begin{proof}
We first prove that the existence of a $k$-mini-hitting set in $H$ implies the existence of a bush in $B_H^*$ of dimension satisfying (a) and (b).
Let $S_{\textsc{mini}}=\{w_1,\ldots,w_q\}$ be a $k$-mini-hitting set
and let $S_i=\{w_1,\ldots,w_i\}$. We know that $q\leq k$
and $|{\cal F}[S_{\textsc{mini}}]| \ge |S_{\textsc{mini}}| + k$.  We define
${\cal E}_i :={\cal F}[S_i]\setminus {\cal F}[S_{i-1}]$ for every $i\geq 2$, and
${\cal E}_1:={\cal F}[S_1].$
Let ${\cal E}_{s_1},\ldots , {\cal E}_{s_r}$ be the subsequence of the sequence ${\cal E}_1,\ldots ,{\cal E}_q$
consisting only of non-empty sets ${\cal E}_i$, and let $b_j=|{\cal E}_{s_j}|$ for each $j\in [r]$.
Let $p$ be the least integer from $[r]$ such that $\sum_{i=1}^{p}b_i\geq k+p$.

Observe that for every $j\in [p]$, the vertex $w_{s_j}$ belongs to each hyperedge of ${\cal E}_{s_j}$. Thus,
the bipartite graph $B_H$ contains a star-forest $F$ of dimension $(b_{1},\ldots ,b_{p})$, such that $p\leq k$,
$b_j\geq 1$ for all $j\in [p]$, and $c:=\sum_{j=1}^{p}b_j\ge p+k.$ Moreover, each star in $F$ has a central vertex in $V.$
By the minimality of $p$, we have $\sum_{j=1}^{p-1}b_j< p-1+k$
and so $b_p\ge c+1-(p+k).$ Thus, the integers $a_j$ defined as follows are positive: $a_j:=b_j$ for every $j\in [p-1]$ and $a_p:=b_p-c+(p+k)$.
Hence, $B_H$ contains a star-forest $F'$ of dimension $(a_{1},\ldots ,a_{p})$, such that each star in $F'$ has a central vertex in $V.$

Thus, all central vertices are in $V$,
$p\leq k$, $a_i\geq 1$ for all $i\in [p]$, and $\sum_{i=1}^{p}a_i=p+k$, which
implies that $B^*_H$ contains, as  a subgraph, a bush with dimension  $(a_1,\ldots, a_p)$ satisfying the conditions above.
%

The construction above relating a $k$-mini-hitting set of $H$ with the required bush of $B_H^*$ can be easily reversed in the following sense:
the existence of a bush of dimension satisfying (a) and (b) in $B^*_H$ implies the existence of
a $k$-mini-hitting set in $H$. Here the triangle ensures that the central vertices are in $V.$  This completes the proof.
\end{proof}

Next we describe a fast randomized algorithm for deciding the existence of a $k$-mini-hitting set using the characterization obtained in
Lemma~\ref{lem:alternatechar}. Towards this we will use a fast randomized algorithm for
the {\sc Subgraph Isomorphism} problem.   In the {\sc Subgraph   Isomorphism} problem we
are given two graphs $F$ and $G$ on $k$ and $n$ vertices, respectively, as an input, and the question is
whether there exists a subgraph of $G$ isomorphic to $F$. Recall that $\tw(G)$ denotes the treewidth of
a graph $G$. We will use the following result.

\begin{theorem} [Fomin {\em et al.}\cite{FominLRSR12}]
\label{thm:subgraphiso}
Let $F$ and $G$ be two graphs on $q$ and $n$ vertices respectively and $\tw(F)\leq t$.
Then, there is a randomized algorithm for the {\sc Subgraph
Isomorphism} problem that runs in expected time $O(2^q(nt)^{t+O(1)})$.
\end{theorem}

Let ${\cal P}_\ell(s)$ be the set of all {\em unordered partitions} of an integer $s$  into $\ell$ parts.
Nijenhuis and Wilf \cite{NijenhuisW78} designed a polynomial delay generation algorithm for partitions of ${\cal P}_\ell(s)$.
Let $p(s)$ be the partition function, i.e., the overall number of partitions of $s$. The
asymptotic behavior of $p(s)$ was first evaluated by
Hardy and Ramanujan in the paper in which they develop the
famous ``circle method.''
\begin{theorem}[Hardy and Ramanujan \cite{HardyR18}]
\label{prop:partfunction}
We have $p(s)\sim e^{\pi\sqrt{\frac{2s}{3}}}/(4s\sqrt{3})$, as $s\to\infty$.
\end{theorem}

This theorem and the  algorithm of Nijenhuis and Wilf~\cite{NijenhuisW78} imply the following:

\begin{proposition}\label{prop:gener}
There is an algorithm of runtime $2^{O(\sqrt{s})}$ for generating all partitions in ${\cal P}_\ell(s)$.
\end{proposition}

Now we are ready to describe and analyze a fast randomized algorithm for deciding the existence of a $k$-mini-hitting set in a hypergraph $H$.
By Lemma~\ref{lem:alternatechar}, it suffices to design and analyze a fast randomized algorithm for deciding the existence of a bush in $B^*_H$ of dimension $(a_1,\ldots,a_p)$ satisfying conditions (a) and (b) of Lemma~\ref{lem:alternatechar}. Our algorithm starts by building $B^*_H$. Then it
considers all possible values of $p$ one by one ($p\in [k]$) and generates all partitions in ${\cal P}_p(p+k)$ using the algorithm of Proposition \ref{prop:gener}. For each such partition $(a_1,\ldots,a_p)$ that satisfies conditions (a) and (b) of Lemma~\ref{lem:alternatechar}, the algorithm of Fomin {\em et al.}\cite{FominLRSR12} mentioned in Theorem \ref{thm:subgraphiso} decides whether $B^*_H$ contains a bush of dimension $(a_1,\ldots,a_p)$. If such a bush exists, we output {\sc Yes} and we output {\sc No}, otherwise.

To evaluate the runtime of our algorithm, observe that the treewidth of any bush
is  $2$ and any bush in Lemma \ref{lem:alternatechar} has at most $3k+3$ vertices. This observation, the algorithm above, Theorem~\ref{thm:subgraphiso} and Proposition \ref{prop:gener} imply the following:

\begin{theorem}
\label{thm:fasterm-khset-purerandom}
There exists a randomized algorithm solving  $(m-k)$-\textsc{Hitting Set} in expected time  \\
$O(8^{k+O(\sqrt{k})}(m+n)^{O(1)})$.
\end{theorem}

This theorem, in turn, implies the following:

\begin{theorem}\label{thm:purerandom-specialinst}
There exists a randomized algorithm solving a transformed special instance of $(\nu(F)+k)$-\textsc{SAT} in expected time
$O(8^{k+O(\sqrt{k})}(m+n)^{O(1)})$.
\end{theorem}

\section{Complete Algorithm, Correctness and Analysis}
\label{sec:complete}
The complete algorithm for an instance $(F, \alpha)$ of $(\nu(F)+k)$-\textsc{SAT} is as follows.

Find a maximum matching $M$ on $B_F$ and let $k=\alpha - |M|$.
If $k\le 0$, return {\sc Yes}.
Otherwise, apply Reduction Rules \ref{pureliteral} to \ref{k2}, whichever is applicable, in
that order and then run the algorithm on the reduced instance and return the answer.
If none of the Reduction Rules apply, then apply Branching Rule \ref{bran1} if possible, to get two instances $(F',\alpha')$ and $(F'',\alpha'')$. Run the algorithm on both instances; if one of them returns {\sc Yes}, return {\sc Yes}, otherwise return {\sc No}. If Branching Rule \ref{bran1} does not apply then we rearrange the formula and attempt to apply Branching Rule \ref{bran2} in the same way.
Finally if $k > 0$ and none of the reduction or branching rules apply, then we have for all variables $x$, $n(x)=1$ and every clause contains at most one positive literal, i.e. $(F, \alpha)$ is a special instance. Then solve the problem by first obtaining the transformed special instance, then the corresponding instance $H^*$ of $(m-k)$-\textsc{Hitting Set} and solving $H^*$ in time
$O((2e)^{2k + O(\log^2 k)} (m+n)^{O(1)})$
as described in Sections \ref{sec:branch} and \ref{sec:m-khs}.

Correctness of all the preprocessing rules and the branching rules follows from Lemmata~\ref{lem:cor1}, \ref{lem:cor2}, \ref{lem:cor3}, \ref{lem:cor4} and \ref{lem:cor5}.

\paragraph{Analysis of the algorithm.}
Let $(F,\alpha)$ be the input instance. Let $\mu(F)=\mu=\alpha -\nu(F)$ be the measure.
We will first show that our preprocessing rules do not increase this measure. Following this, we will prove a lower bound on the
decrease in the measure occurring as a result of the branching, thus allowing us to bound the running time of the algorithm in
terms of the measure $\mu$. For each case, we let $(F',\alpha')$ be the instance resulting by the application of the rule or
branch. Also let $M'$ be a maximum matching of $B_{F'}$.\\

\noindent{\bf Reduction Rule~\ref{pureliteral}: } We consider the case when $n(x)=0$; the other case when
$n(\bar{x})=0$ is analogous. We know that $\alpha'=\alpha-n(\bar{x})$ and $\nu(F')\geq \nu(F)-n(\bar{x})$ as removing $n(\bar{x})$ clauses
can only decrease the matching size by $n(\bar{x})$. This implies that $\mu(F)-\mu(F')=\alpha -\nu(F)-\alpha'+\nu(F')=(\alpha-\alpha')+(\nu(F')-\nu(F))
\geq n(\bar{x})-n(\bar{x})$. Thus,  $\mu(F')\leq \mu(F)$.

\medskip

\noindent{\bf Reduction Rule~\ref{twooccurrence}: } We know that $\alpha'=\alpha-1$. We show that $\nu(F')\geq \nu(F)-1$. In this
case we remove the clauses $c'$ and $c''$ and add $c^*=(c'-x)\cup (c''-\bar{x})$. We can obtain a matching of size $\nu(F)-1$ in $B_{F'}$ as
follows. If at most one of the $c'$ and $c''$ is the end-point of some matching edge in $M$ then removing that edge gives a matching of size
$\nu(F)-1$ for $B_{F'}$. So let us assume that some edges $(a,c')$ and $(b,c'')$ are in $M$. Clearly, either $a\neq x$ or $b\neq x$. Assume
$a\neq x$. Then $M\setminus \{(a,c'),(b,c'')\}\cup \{(a,c^*)\}$ is a matching of size $\nu(F)-1$ in $B_{F'}$. Thus, we conclude that
$\mu(F')\leq \mu(F)$.

\medskip

\noindent{\bf Reduction Rule~\ref{k1}: }
The proof
is the same as in the case of Reduction Rule~\ref{pureliteral}.

\medskip

\noindent{\bf Reduction Rule~\ref{k2}: } The proof that $\mu(F')\leq \mu(F)$ in the case when $F[S]$ is satisfiable
is the same as in the case of Reduction Rule~\ref{pureliteral} and in the case when $F[S]$ is not satisfiable is the same as in the case of Reduction Rule~\ref{twooccurrence}.

\medskip

\noindent{\bf Branching Rule \ref{bran1}: }
Consider the case when we set $x=\textsc{true}$. In this case, $\alpha^\prime=\alpha-n(x)$. Also, since no reduction rules are applicable we have
that $F$ is $2$-expanding.
Hence, $\nu(F)=|V(F)|$.  We will show that in $(F',\alpha')$ the matching size will remain at least
$\nu(F)-n(x)+1$ ($= |V(F)|-n(x)+1 = |V(F')|- n(x) +2$.) This will imply that $\mu(F')\leq \mu(F)-1$.
By Lemma \ref{lem:dless} and the fact that $n(x)-2 \ge 0$, it suffices to show that
in $B'=B_{F'}$,  every subset $S\subseteq V(F')$,
$|N_{B'}(S)|\geq |S| -(n(x)-2)$.
%
%
%
%
The only clauses that have been removed by the simplification process after
setting  $x=\textsc{true}$ are those where $x$ appears positively and the singleton clauses $(\bar{x})$. Hence, the only edges of $G[S\cup N_B[S]]$
that are missing in $N_{B'}(S)$ from $N_{B}(S)$ are those corresponding to clauses that contain $x$ as a pure literal and some variable in $S$.
Thus, $|N_{B'}(S)|\geq |S|+2 -n(x) = |S| -(n(x)-2)$ (as $F$ is $2$-expanding).

The case when we set $x=\textsc{false}$ is similar to the case when we set $x=\textsc{true}$. Here, also we can show that $\mu(F')\leq \mu(F)-1$.
Thus, we get two instances, with each instance $(F', \alpha')$ having $\mu(F')\le \mu(F)-1$.

\medskip

\noindent{\bf Branching Rule \ref{bran2}: } The analysis here is the same as for Branching Rule \ref{bran1} and
again we get
two instances with $\mu(F')\le \mu(F)-1$.

\medskip



We therefore have a depth-bounded search tree of size of depth at most $\mu = \alpha - \nu(F) = k$, in which any branching splits an instance into two instances. Thus, the search tree has at most $2^k$ instances. As each reduction and branching rule takes polynomial time, every rule decreases the number of variables, the number of clauses, or the value of $\mu$,
and an instance to which none of the rules apply can be solved in time
$O((2e)^{2\mu} \mu^{O(\log \mu)}(m+n)^{O(1)})$ (by Theorem \ref{thm:fpt-specialinst}),
we have by induction that any instance can be solved in time
$$O(2\cdot(2e)^{2(\mu-1)} (\mu-1)^{O(\log (\mu-1))}(m+n)^{O(1)})= O((2e)^{2\mu} \mu^{O(\log \mu)}(m+n)^{O(1)}).$$
Thus
the total running time of the algorithm is at most
$O((2e)^{2k + O(\log^2 k)}(n+m)^{O(1)})$.
Applying Theorem \ref{thm:purerandom-specialinst} instead of Theorem \ref{thm:fpt-specialinst}, we conclude that
$(\nu(F)+k)$-\textsc{SAT} can be solved in expected time $O(8^{k + O(\sqrt{k})}(n+m)^{O(1)})$.
Summarizing, we have the following:

\begin{theorem} \label{main}
There are algorithms solving $(\nu(F)+k)$-\textsc{SAT} in time  \\
$O((2e)^{2k + O(\log^2 k)} (n+m)^{O(1)})$ or expected time $O(8^{k + O(\sqrt{k})}(n+m)^{O(1)})$.
\end{theorem}

\section{Hardness of Kernelization}
\label{sec:nokern}

In this section, we show that \nksat{} does not have a polynomial-size kernel, unless coNP $\subseteq$ NP/poly.
To do this, we use the concept of a \emph{polynomial parameter transformation} \cite{BTY09,DLS09}:
 Let $L$ and $Q$ be parameterized problems. We say a polynomial time computable function $f: \Sigma^* \times \mathbb{N} \rightarrow \Sigma^* \times \mathbb{N}$ is a \emph{polynomial parameter transformation} from $L$ to $Q$ if there exists a polynomial $p: \mathbb{N} \rightarrow \mathbb{N}$ such that for any $(x,k)\in \Sigma^* \times \mathbb{N}, (x,k) \in L$ if and only if $f(x,k)=(x',k') \in Q$, and $k' \le p(k)$.

\begin{lemma}{\rm \cite[Theorem 3]{BTY09}}\label{ppt}
Let $L$ and $Q$ be parameterized problems, and suppose that $L^c$ and
$Q^c$ are the derived classical problems\footnote{The parameters of $L$ and $Q$ are no longer parameters in $L^c$ and
$Q^c$; they are part of input.}. Suppose that $L^c$ is NP-complete, and $Q^c \in$ NP.
 Suppose that $f$ is a polynomial parameter transformation from $L$
to $Q$. Then, if $Q$ has a polynomial-size kernel, then $L$ has a polynomial-size kernel.
\end{lemma}

The proof of the next theorem is similar to the proof of Lemma~\ref{lem:newtransformation}.
\begin{theorem}
\nksat{} has no polynomial-size kernel, unless coNP $\subseteq$ NP/poly.
\end{theorem}

\begin{proof}
By \cite[Theorem 3]{GutinJY11}, there is no polynomial-size kernel for the problem of deciding whether a hypergraph $H$ has a hitting set of size $|E(H)|-k$, where $k$ is the parameter unless coNP $\subseteq$ NP/poly.
We prove the theorem by a polynomial parameter reduction from this problem. Then the theorem follows from Lemma \ref{ppt}, as  ~\nksat ~ is NP-complete.

Given a hypergraph $H$ on $n$ vertices, construct a CNF formula $F$ as follows. Let the variables of $F$ be the vertices of $H$. For each variable $x$, let the unit clause $(x)$ be a clause in $F$. For every edge $e$ in $H$, let $c_e$ be the clause containing the literal $\bar{x}$ for every $x \in E$.
Observe that $F$ is matched, and that $H$ has a hitting set of size $|E(H)|-k$ if and only if $\sat(F) \ge n+k$.
\end{proof}
\section{Conclusion}
\label{sec: concl}
We have shown that for any CNF formula $F$, it is fixed-parameter tractable to decide if $F$ has a satisfiable subformula containing $\alpha$ clauses, where $\alpha - \nu(F)$ is the parameter. Our result implies fixed-parameter tractability for the problem of deciding satisfiability of $F$ when $F$ is variable-matched and $\delta(F)\le k$, where $k$ is the parameter.
In addition, we show that the problem does not have a polynomial-size kernel unless coNP $\subseteq$ NP/poly.

Clearly, parameterizations of {\sc MaxSat} above $m/2$ and $\nu(F)$ are ``stronger'' than the standard parameterization (i.e., when the parameter is the size of the solution). Whilst the two non-standard parameterizations have smaller parameter than the standard one, they are incomparable to each other as for some formulas $F$, $m/2<\nu(F)$ (e.g., for variable-matched formulas with $m<2n$) and for some formulas $F$, $m/2>\nu(F)$ (e.g., when $m>2n$). Recall that Mahajan and Raman \cite{MR99} proved that {\sc MaxSat} parameterized above $m/2$ is fixed-parameter tractable. This result and our main result imply that {\sc MaxSat} parameterized above $\max\{m/2,\nu(F)\}$ is fixed-parameter tractable: if $m/2>\nu(F)$ then apply the algorithm of \cite{MR99}, otherwise apply our algorithm.

If every clause of a formula with $m$ clauses contains exactly two literals then it is well known that we can satisfy at least $3m/4$ clauses. From this, and by applying Reduction Rules \ref{pureliteral} and \ref{twooccurrence}, we can get a linear kernel for this version of the \nksat ~ problem.
It would be nice to see whether a linear or a polynomial-size kernel exists for the \nksat ~ problem if every clause has exactly $r$ literals.

\medskip

\paragraph{Acknowledgment}
This research was partially supported by an International Joint grant of the Royal Society.

\end{document}